\documentclass[11pt]{article}
\usepackage[margin=1in]{geometry}
\usepackage{graphicx}
\usepackage{amsmath}
\usepackage{amsthm}
\usepackage{physics}
\usepackage{amssymb}
\usepackage{xcolor}
\usepackage{hyperref}
\usepackage{mathtools}
\usepackage{comment}
\usepackage{subcaption}
\usepackage{enumitem}
\usepackage[normalem]{ulem}

\hypersetup{
    colorlinks,
    linkcolor={red!50!black},
    citecolor={blue!50!black},
    urlcolor={blue!80!black}
}

\theoremstyle{plain}
\newtheorem{theorem}{Theorem}[section]
\newtheorem{lemma}[theorem]{Lemma}
\newtheorem{proposition}[theorem]{Proposition}
\newtheorem{corollary}[theorem]{Corollary}

\theoremstyle{definition}
\newtheorem{definition}[theorem]{Definition}

\theoremstyle{remark}

\newtheorem{remark}[theorem]{Remark}
\newtheorem{assumption}{Assumption}

\numberwithin{equation}{section}

\usepackage{dsfont}
\newcommand{\indicator}[1]{\mathds{1}_{#1}}
\newcommand{\set}[1]{\left\{#1\right\}}

\renewcommand{\P}{\mathbb{P}}
\newcommand{\R}{\mathbb{R}}

\newcommand{\Z}{\mathbb{Z}}
\newcommand{\bG}{\mathbb{G}}
\newcommand{\bK}{\mathbb{K}}
\newcommand{\bL}{\mathbb{L}}

\newcommand{\cF}{\mathcal{F}}

\newcommand{\cL}{\mathcal{L}}

\newcommand{\cT}{\mathcal{T}}

\newcommand{\rect}[2][]{
  \if\relax\detokenize{#1}\relax
    \mathrm{R}_{#2}
  \else
    \mathrm{R}_{#2,#1}
  \fi
}

\newcommand{\supp}[1]{\mathrm{supp}(#1)}

\newcommand{\Area}[1]{\mathrm{Area}(#1)}

\newcommand{\intr}[1]{\mathrm{int}(#1)}

\newcommand{\onion}[2][n]{#2^{(#1)}}

\let\emptyset\varnothing

\allowdisplaybreaks

\title{Unified criteria for crystallization in hard-core lattice systems with applications to polyomino fluids and multi-component mixtures}
\author{Qidong He}
\date{}

\begin{document}

\maketitle

\begin{abstract}
    We present a unified extension of two sets of criteria for high-fugacity crystallization in hard-core lattice systems developed previously by Jauslin, Lebowitz, and the author.
    Our new criterion is formulated in terms of the existence of a volume allocation rule with desirable optimization and screening properties, in analogy to the scoring function constructed in Hales' proof of the Kepler conjecture.
    Notably, our result applies to a large class of polyomino models with discrete rotational degrees of freedom and their chiral mixtures, as well as multi-component mixtures featuring several geometrically distinct particle shapes.
    The proof uses a recent systematic extension of Pirogov--Sinai theory to systems with infinite interactions by Mazel--Stuhl--Suhov.
\end{abstract}

\section{Introduction}

Recently, Mazel--Stuhl--Suhov~\cite{mazel2024pirogov} systematically extended the version of Pirogov--Sinai theory~\cite{Pirogov1975} due to Zahradn\'ik's~\cite{zahradnik1984alternate} to lattice systems with infinite interactions.
In this note, we use their result to generalize the criteria developed by Jauslin, Lebowitz, and the author in~\cite{Jauslin2018,he2024high} for high-fugacity crystallization in hard-core lattice systems.
We assume that the reader is familiar with the main ideas of Pirogov--Sinai theory and refer to~\cite{friedli2017statistical,helmuth2019algorithmic,cannon2024pirogov} for pedagogical references and recent developments.
We note that although~\cite{Jauslin2018,he2024high} were also based on Pirogov--Sinai theory, the implementation therein relied on very specific geometric properties of the models and led to various technical assumptions that limited their applicability, but which can be overcome using the framework of~\cite{mazel2024pirogov}, as we now explain.

One example of such an assumption is the requirement that the close-packings of the particles be related by isometries of the underlying lattice~\cite[Definition 1]{Jauslin2018},~\cite[Assumption~1.3]{he2024high}.
This condition was essential to ensuring the cancellation of the bulk terms in ratios of partition functions with different boundary conditions.
The reason was that contours in~\cite{Jauslin2018,he2024high} were constructed on the level of particles using a geometric definition of correctness that is highly sensitive to the local packing structure.
As Mazel--Stuhl--Suhov pointed out in a general context~\cite{mazel2024pirogov}, a more elegant approach is to first coarse-grain the system on the scale of a \textit{supercell} (see Section~\ref{sec:the assumptions}) common to all periodic close-packings, and then construct contours on the level of supercells using the usual notion of correctness---that a supercell and all its neighbors agree with the same periodic ground state.
This approach automatically guarantees the desired cancellation and removes the need for the isometry condition as well as~\cite[Assumption~1.5]{he2024high}, which addressed similar geometric issues.
Thus, it paves the way for treating models featuring several translationally distinct particle shapes, and those admitting several non-congruent crystal structures; see Figure~\ref{fig:Z pentomino crystals}.

\begin{figure}[t]
    \centering
    \begin{subfigure}{0.3\columnwidth}
        \centering
        \includegraphics[scale=0.3]{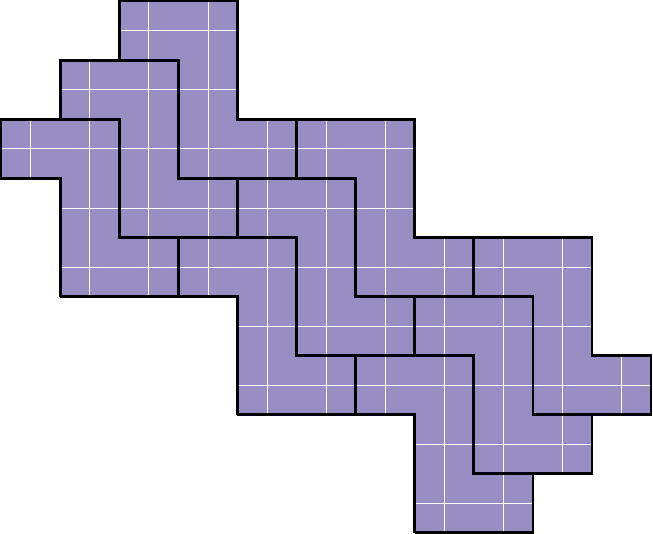}
        \caption{}
    \end{subfigure}
    \begin{subfigure}{0.3\columnwidth}
        \centering
        \includegraphics[scale=0.3]{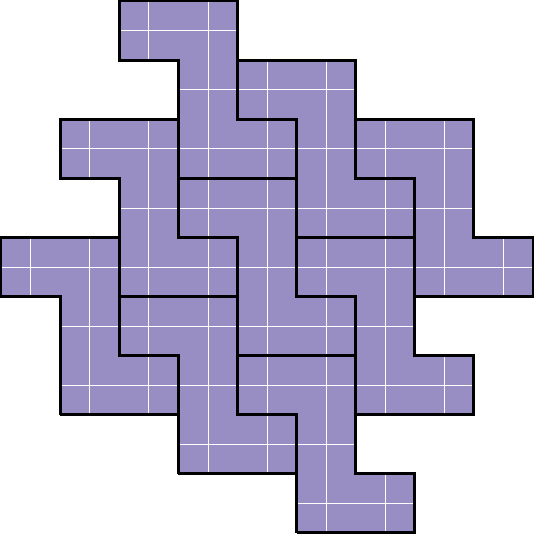}
        \caption{}
    \end{subfigure}
    \begin{subfigure}{0.32\columnwidth}
        \centering
        \includegraphics[scale=0.3]{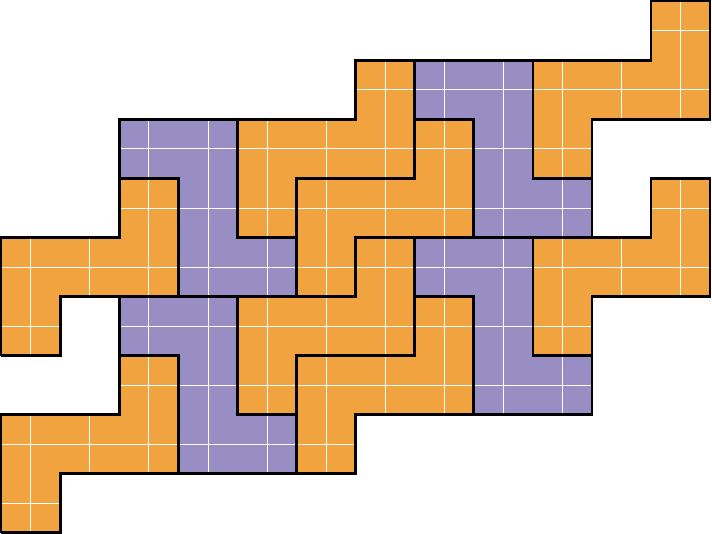}
        \caption{}
    \end{subfigure}

    \vspace{0.5cm}

    \begin{subfigure}{0.3\columnwidth}
        \centering
        \includegraphics[scale=0.3]{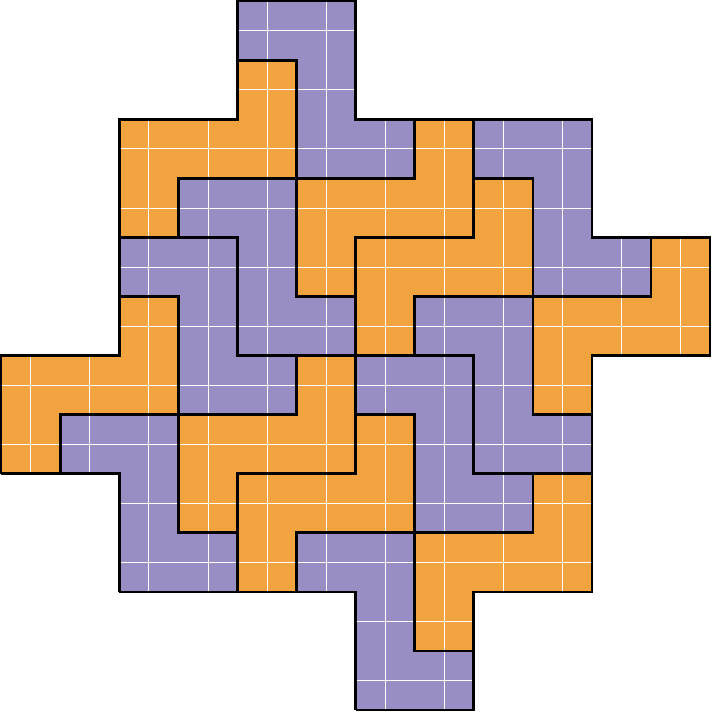}
        \caption{}
    \end{subfigure}
    \begin{subfigure}{0.3\columnwidth}
        \centering
        \includegraphics[scale=0.3]{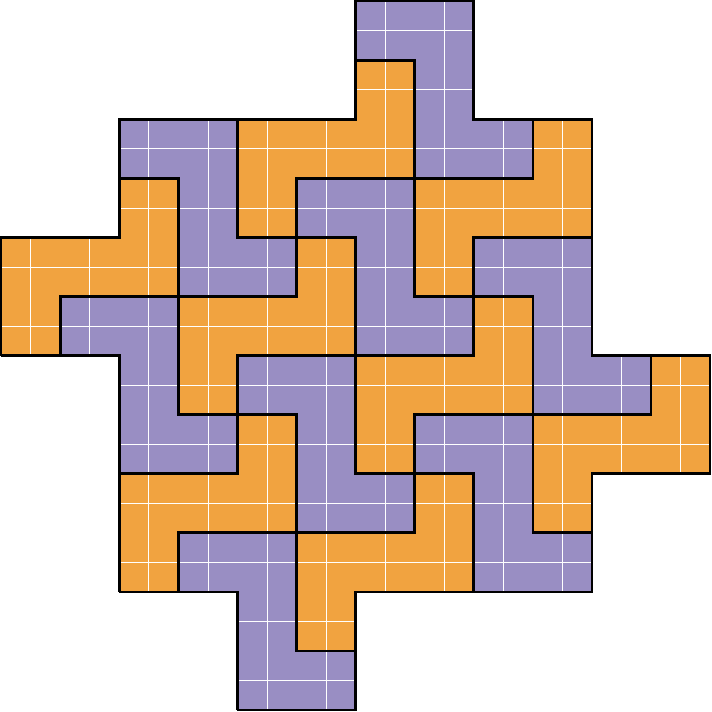}
        \caption{}
    \end{subfigure}
    \begin{subfigure}{0.32\columnwidth}
        \centering
        \includegraphics[scale=0.3]{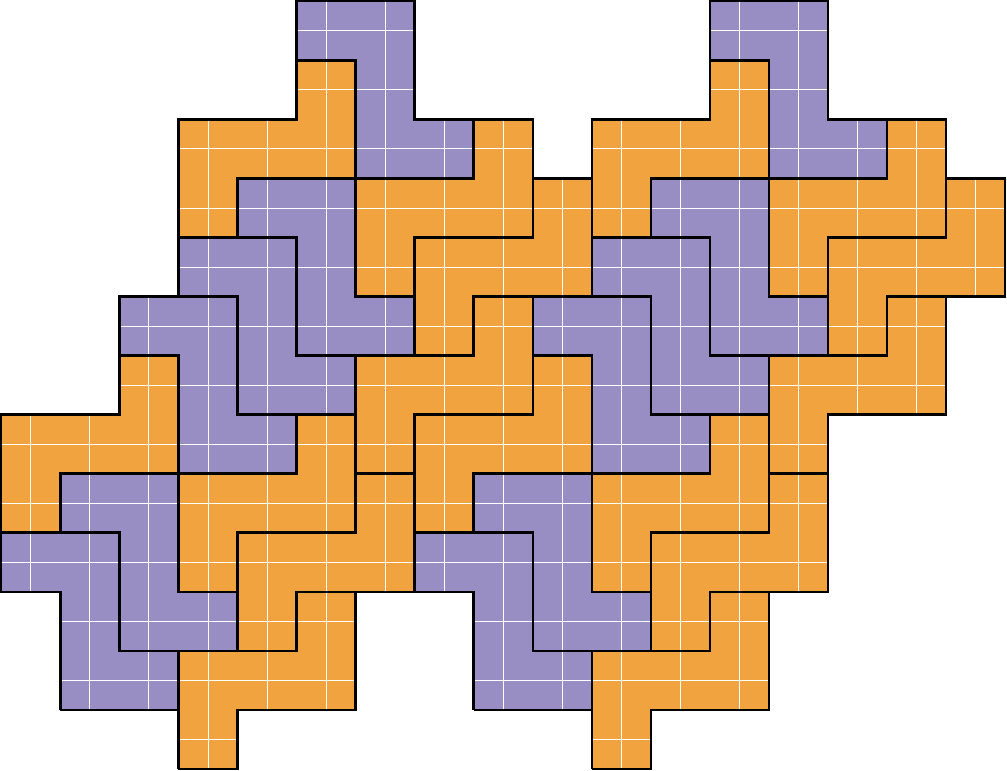}
        \caption{}
    \end{subfigure}
    \caption{The only crystal (tiling) structures of the chiral $Z$-pentomino on the square lattice up to translations and rotations~\cite[Chapter~6]{martin1991polyominoes}}
    \label{fig:Z pentomino crystals}
\end{figure}

Another point of generalization comes from the verification of the Peierls condition.
In~\cite{Jauslin2018}, the close-packings were required to cover every site of the lattice, and the Peierls condition held because any contour must contain at least a proportional number of uncovered sites.
With applications to non-tiling models in mind,~\cite{he2024high} instead used the discrete Voronoi tessellation to identify deviations from close-packings, defined as regions featuring an increase in the volumes of the (discrete) Voronoi cells, and the Peierls condition followed by assuming that the close-packings simultaneously minimized the volume of all the Voronoi cells.
In their series of works on crystallization in disk and sphere models on lattices~\cite{mazel2019high,mazel2020hard,mazel2021kepler,mazel2023hard,mazel2025high}, Mazel--Stuhl--Suhov identified similar minimization properties of \textit{local volumes} by considering other common tessellations, such as the Delaunay triangulation with respect to the geometric centers of the particles.
Their insight is proof enough that the formulation of the minimization property need not be tied to a specific measure of local volume.
Building on this idea, we formulate our new criterion in terms of the existence of a \textit{volume allocation} rule with desirable properties to make it versatile in applications.

An archetypal example of a system covered by the present note is that of the chiral $Z$-pentomino on the square lattice with discrete rotational degrees of freedom, which has been studied in the computational chemistry literature~\cite{barnes2009structure,barnes2010development} and is known to admit a remarkable six non-congruent crystal structures; see Figure~\ref{fig:Z pentomino crystals}.
A plethora of other polyomino models with discrete rotational degrees of freedom, as well as chiral mixtures thereof, also fall within our framework.
Examples may be found in~\cite[Chapters~7--9]{martin1991polyominoes}; see Figure~\ref{fig:other shapes} for a few and Section~\ref{sec:polyominoes} for discussion.
Finally, in Section~\ref{sec:diamond-octagon mixtures}, we give an example of a multi-component mixture, featuring diamonds and octagons, that crystallizes at fine-tuned chemical potentials and low temperatures, with one of two possible crystal structures given by the well-known truncated square tiling of the plane; see Figure~\ref{fig:truncated square tiling}.

\begin{remark}
    The initial motivation for this work came from the grand-canonical Monte Carlo simulations of polyomino fluids reported in the doctoral thesis of Barnes~\cite{barnes2010development}.
    We strongly encourage the reader to examine the beautiful pictures therein, specifically~\cite[Figures~5.13--16]{barnes2010development} which showcase the structural evolution of the $Z$-pentomino fluid across the crystallization transition.
    Beyond the transition point, the simulated system on a torus settles into a polycrystalline state, where regions with distinct crystal structures depicted in Figure~\ref{fig:Z pentomino crystals} are separated by long-lived domain boundaries.
    Corollary~\ref{cor:polyominoes} effectively confirms these findings by proving that the crystal structures observed in~\cite[Figures~5.14--16]{barnes2010development} are not artifacts of the finite torus geometry, but indeed correspond to stable infinite-volume phases of the system.
\end{remark}

\begin{figure}[t]
    \centering
    \captionsetup[subfigure]{justification=centering}
    \begin{subfigure}[b]{0.15\columnwidth}
        \centering
        \includegraphics[scale=0.5]{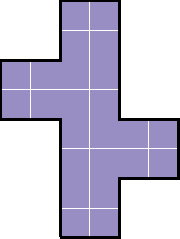}
        \caption{$2$-poic}
    \end{subfigure}
    \begin{subfigure}[b]{0.23\columnwidth}
        \centering
        \includegraphics[scale=0.5]{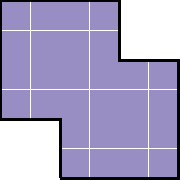}
        \caption{$2$-morphic, $2$-poic}
    \end{subfigure}
    \begin{subfigure}[b]{0.23\columnwidth}
        \centering
        \includegraphics[scale=0.5]{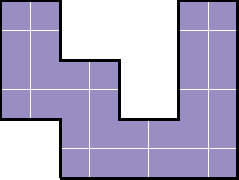}
        \caption{$4$-morphic, $4$-poic}
    \end{subfigure}
    \begin{subfigure}[b]{0.3\columnwidth}
        \centering
        \includegraphics[scale=0.5]{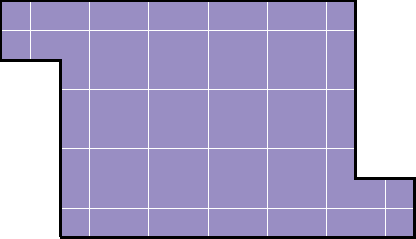}
        \caption{$3$-morphic, $2$-poic}
    \end{subfigure}
    \caption{Other examples of polyominoes that tile $\Z^2$ in multiple non-congruent ways when rotations, or reflections as well, are allowed~\cite[Problems~7.15, 8.10, 9.3, 9.8]{martin1991polyominoes}}
    \label{fig:other shapes}
\end{figure}

\section{Volume allocations and crystallization}

\subsection{The model}
\label{sec:the model}

Let $B$ be a non-degenerate $d\times d$ matrix, $\bL:=B\cdot\Z^d$, and $\bG$ be the union of finitely many disjoint translations of $\bL$.
Let $T_1,\dots,T_k$ be non-empty, bounded subsets of $\R^d$, which we refer to as \textbf{tiles}.
Fix a \emph{unit} vector $\vec{\mu}=(\mu_1,\dots,\mu_k)\in\R^k$, whose components are interpreted as the chemical potentials of the tiles.
An \textbf{admissible configuration} is a function $\omega:\bG\to\set{0,1,\dots,k}$ such that 
\begin{equation}
    \intr{u+T_{\omega(u)}}\cap\intr{v+T_{\omega(v)}}=\emptyset\quad\text{ for every distinct }u,v\in\bG\text{ such that }\omega(u),\omega(v)\ne 0,
\end{equation}
where $\intr{\cdot}$ denotes the interior of a set.
We identify an admissible configuration $\omega$ with the set 
\begin{equation}
    X=X(\omega):=\set{(v,\omega(v))\mid v\in\bG,\omega(v)\ne 0},
\end{equation}
and write $X_\Lambda:=X\cap\Lambda$ for $\Lambda\subset\R^d$, and $T_x:=v+T_{\omega(v)}$, $\mu_x:=\mu_{\omega(v)}$ for $x=(v,\omega(v))\in X(\omega)$.
Let $\Omega$ be the set of admissible configurations and equip it with the (metrizable) topology of local convergence~\cite[Section~4.1]{georgii2011gibbs}: Under this topology, two admissible configurations are \emph{close} if they agree on a large ball centered at the origin.

The formal Hamiltonian of an admissible configuration $X\in\Omega$ is defined as
\begin{equation}
    \label{eqn:formal Hamiltonian}
    H(X):=-\sum_{x\in X}\mu_x+\sum_{\set{x,y}\subset X}\infty\cdot\indicator{\intr{T_x}\cap\intr{T_y}\ne\emptyset},
\end{equation}
where we interpret $\infty\cdot 0:=0$.
The Gibbs measure(s) of the model~\eqref{eqn:formal Hamiltonian} at inverse temperature $\beta$ are defined via the usual Dobrushin--Lanford--Ruelle prescription~\cite{biskup2009reflection,georgii2011gibbs}.
In this paper, we assume that $\mu_1,\dots,\mu_k>0$ and focus on the regime $\beta\gg 1$.

\subsection{The assumptions}
\label{sec:the assumptions}

\begin{definition}
    \label{def:volume allocation}
    Let $S\subset\R^d$ be an $\bL$-invariant space with an $\bL$-invariant measure $\lambda$ (which need not be the Lebesgue measure).
    A \textbf{volume allocation} is a family of measurable functions $\phi_x(\cdot\mid X):S\to[0,1]$, defined for every non-empty admissible configuration $X\in\Omega$ and $x\in X$, satisfying the following properties:
    \begin{enumerate}
        \item \textit{Translation covariance}: For every $t\in\bL$, 
        \begin{equation}
            \label{eqn:translation covariance}
            \phi_{t+x}(t+\cdot\mid t+X)=\phi_x(\cdot\mid X)\quad\lambda\text{-a.e.}
        \end{equation}
        \item \textit{Lower semi-continuity}: For all sequences of admissible configurations $(X_n)_{n\ge 1}$ converging to $X$,
        \begin{equation}
            \label{eqn:lower semi-continuity}
            \liminf_{n\to\infty}\phi_x(\cdot\mid X_n)\ge \phi_x(\cdot\mid X)\quad\lambda\text{-a.e.}
        \end{equation}
        \item \textit{Partition of unity}: 
        \begin{equation}
            \label{eqn:partition of unity}
            \sum_{x\in X}\phi_x(\cdot\mid X)=1\quad\lambda\text{-a.e.}
        \end{equation}
    \end{enumerate}
\end{definition}

Given a volume allocation $\set{\phi_x(\cdot\mid X)}$, the \textbf{allocated volume} to $x\in X$ in a non-empty admissible configuration $X\in\Omega$ is defined as
\begin{equation}
    v(x\mid X):=\int_S\lambda(\dd{s})\phi_x(s\mid X).
\end{equation}
We introduce the following shorthand notation: For every non-empty admissible configuration $X\in\Omega$ and $X'\subset X$, write
\begin{equation}
    v(X'\mid X):=\sum_{x\in X'}v(x\mid X).
\end{equation}

\begin{assumption}
    \label{asm:optimization}
    There exists a volume allocation $\set{\phi_x(\cdot\mid X)}$ and a constant $p_\ast>0$ such that:
    \begin{enumerate}[label=(\alph*),ref=\theassumption(\alph*)]
        \item \label{itm:minorization} \textit{Local optimization}: For every non-empty admissible configuration $X\in\Omega$ and $x\in X$, 
        \begin{equation}
            p_\ast\, v(x\mid X)\ge \mu_x.
        \end{equation}
        \item \label{itm:global optimization} \textit{Global optimization}: There exist only finitely many non-empty admissible configurations $\overline{X}\in\Omega$, referred to as \textbf{perfect configurations}, such that 
        \begin{equation}
            p_\ast\, v(x\mid\overline{X})=\mu_x\quad\text{for every }x\in \overline{X}.
        \end{equation}
    \end{enumerate}
\end{assumption}

To state the next assumption, we first introduce some terminology.
A non-empty admissible configuration is \textbf{periodic} if it is invariant under translations by a full-rank sublattice of $\bL$.
Let $\bK\subset\bL$ be a full-rank sublattice.
A \textbf{supercell} of $\bK$ is the region $A\cdot[0,1)^d\subset\R^d$ for a full-rank sublattice $A\cdot\Z^d\subset\bK$.
Conversely, given a non-degenerate parallelepiped $G=A\cdot[0,1)^d$, define $\cL^{G}:=A\cdot\Z^d$ as the \textbf{grid} of $G$, and denote by $\cT^G:=\cL^G\cdot G$ the set of translations of $G$ under $\cL^G$, and $\cF^G:=\sigma(\cT^G)$ the $\sigma$-algebra generated by $\cT^G$.
A supercell of a periodic configuration $X$ is a supercell of the largest sublattice of $\bL$ under which $X$ is invariant.
A \textbf{fundamental parallelepiped} of $\bK$ is a supercell of $\bK$ of minimum volume.

By Proposition~\ref{prop:perfect iff periodic ground state}, a perfect configuration is necessarily periodic.
Let $F\subset\R^d$ be a supercell common to all perfect configurations, and abbreviate $\cF:=\cF^F$.
Given $\Lambda\in\cF$, let $\partial_F\Lambda$ denote the smallest $\cF$-measurable subset of $\Lambda$ whose closure intersects the closure of $\Lambda^c$.
Given a parallelepiped $G\subset\R^d$, for $n\ge 0$, let $\onion[n]{G}\subset\R^d$ be the union of $(2n+1)^d$ copies of $G$ arranged in a parallelepiped centered at $G$, such that there are exactly $2n+1$ copies of $G$ along each direction.

Under the coarse-graining scheme of~\cite{mazel2024pirogov}, one chooses a supercell $F$ common to all (finitely many) \emph{periodic ground states} (see Proposition~\ref{prop:perfect iff periodic ground state}) of the model and \textit{larger than the interaction radius} in the following sense: If $\Lambda,\Lambda'\in\cF$ have disjoint closures and $X_\Lambda,X'_{\Lambda'}$ are admissible configurations, then
\begin{equation}
    \label{eqn:interaction screening}
    H(X_\Lambda\mid X'_{\Lambda'})=H(X'_{\Lambda'}\mid X_\Lambda)=0.
\end{equation} 
We refer to the elements of $\cT^F$ as \textbf{cells}.
Thus, the microscopic hard-core interaction between tiles can be embedded in a mesoscopic interaction between cells with non-disjoint closures, yielding an equivalent theory on $\cL^F\cong\Z^d$ with first- and second-nearest-neighbor interactions and a much larger single spin space.
Then, the notion of \textbf{$\overline{X}$-correctness} is defined in the usual way for every periodic ground state $\overline{X}\in\Omega$: A cell $K\in\cT^F$ is $\overline{X}$-correct with respect to an admissible configuration $X$ if $X_{\onion[1]{K}}=\overline{X}_{\onion[1]{K}}$.

\begin{assumption}
    \label{asm:supercell}
    There exists a choice of a supercell $F$ common to all \textit{perfect} configurations (see Proposition~\ref{prop:perfect iff periodic ground state}) and larger than the interaction radius such that:
    \begin{enumerate}[label=(\alph*),ref=\theassumption(\alph*)]
        \item \label{itm:volume screening} \textit{Allocation screening by a perfect configuration}: If $\overline{X}$ is a perfect configuration and $F$ is $\overline{X}$-correct with respect to an admissible configuration $X$, then for every $x\in X_F$,
        \begin{equation}
            \phi_x(\cdot\mid X)=\phi_x(\cdot\mid\overline{X})\quad\lambda\text{-a.e.}
        \end{equation}
    \end{enumerate}
\end{assumption}

\begin{remark}
    The properties imposed on the volume allocation, both in Definition~\ref{def:volume allocation} and Assumptions~\ref{asm:optimization} and~\ref{asm:supercell}, are inspired by those of the scoring function constructed in Hales' proof of the Kepler conjecture~\cite{hales2005proof}, which was later adapted by Mazel--Stuhl--Suhov~\cite{mazel2021kepler,mazel2023hard} to study sphere packings on $\Z^3$.
    Fundamentally, all of these approaches solve the crystallization problem by distributing the ambient space among the particles, thereby reducing the difficult problem of achieving global density maximization to the simpler problem of ensuring local volume minimization.
    Concretely, we point out the following specific analogies:
    \begin{enumerate}
        \item Translation covariance~\eqref{eqn:translation covariance} is explicitly noted in~\cite[Remark~2.11]{hales2005proof}.
        \item Hales' scoring function is in fact continuous~\cite[Lemma~5.9]{hales2005proof}, while lower semi-continuity~\eqref{eqn:lower semi-continuity} is sufficient for our purposes.
        \item The optimization properties in Assumption~\ref{asm:optimization} correspond to~\cite[Theorem~6.1, Corollary~6.3]{hales2005proof}.
        \item The formulation of the screening property in Assumption~\ref{asm:supercell} are based on that of~\cite[Property~9.(IV)]{mazel2021kepler}, which in turn refers to~\cite[Lemma~5.10, Theorem~5.11]{hales2005proof}.
    \end{enumerate}
\end{remark}

\subsection{Crystallization}

In this subsection, we use the framework of~\cite{mazel2024pirogov} to prove high-fugacity crystallization in models satisfying Assumptions~\ref{asm:optimization} and~\ref{asm:supercell}.
The proof consists of two main steps: identifying the periodic ground states, and verifying the Peierls condition.
Each is carried out in a separate sub-subsection below.

\begin{theorem}
    \label{thm:main}
    Under Assumptions~\ref{asm:optimization} and~\ref{asm:supercell}, there exists a universal constant $\beta_0>0$ such that the following holds.
    Index the perfect configurations by $q$.
    For all $\beta\ge \beta_0$ and every perfect configuration $\overline{X}^q$, the model~\eqref{eqn:formal Hamiltonian} admits an $\cL^F$-invariant, extremal Gibbs measure $\P_\beta^q$ such that for every perfect configuration $\overline{X}^{q'}$,
    \begin{equation}
        \P_\beta^q(F\text{ is }q'\text{-correct})
        =\begin{cases}
            1+\order{z^{-1}} & \text{if }q'=q \\
            \order{z^{-1}} & \text{otherwise}
        \end{cases}
        ,
    \end{equation}
    where $z:=\exp{\beta\min_{1\le i\le k}\mu_i}$.
\end{theorem}

\begin{remark}
    Further properties of the Gibbs measures $\P_\beta^q$ and the convergence of an expansion for the pressure of the system in inverse powers of the fugacities $z_i:=e^{\beta\mu_i}$, $1\le i\le k$, are readily extracted from~\cite[Corollary]{mazel2024pirogov}.
\end{remark}

\subsubsection{Periodic ground states}

In this sub-subsection, we identify the periodic ground states of the system as precisely the perfect configurations under Assumptions~\ref{asm:optimization} and~\ref{asm:supercell}.

\begin{proposition}
    \label{prop:perfect iff periodic ground state}
    Every perfect configuration $\overline{X}$ is periodic and a \textbf{ground state} of the system, i.e., for every bounded $\Lambda\subset\bG$ and admissible configuration $X$ coinciding with $\overline{X}$ on $\Lambda^c$,
    \begin{equation}
        H(X_\Lambda\mid \overline{X}_{\Lambda^c})-H(\overline{X}_\Lambda\mid \overline{X}_{\Lambda^c})\ge 0.
    \end{equation}
    Conversely, every periodic ground state is a perfect configuration.
\end{proposition}

First, we note that Properties~\eqref{eqn:translation covariance} and~\eqref{eqn:lower semi-continuity} of the volume allocation extend straightforwardly to the allocated volume, the latter as a consequence of Fatou's lemma.

\begin{lemma}
    \label{lem:translation invariance and lower semicontinuity}
    The allocated volume is translation invariant and lower semi-continuous, i.e., for all non-empty admissible configurations $X\in\Omega$, $x\in X$, and $t\in\bL$,
    \begin{equation}
        v(t+x\mid t+X)=v(x\mid X),
    \end{equation}
    and for all sequences of admissible configurations $(X_n)_{n\ge 1}$ converging to $X$,
    \begin{equation}
        \liminf_{n\to\infty}v(x\mid X_n)\ge v(x\mid X).
    \end{equation}
\end{lemma}

Next, we show that two non-empty periodic configurations have the same allocated volume within a common supercell.

\begin{lemma}
    \label{lem:common supercell equal volume}
    If $X,Y$ are non-empty periodic configurations and $G$ is a supercell common to both $X$ and $Y$, then
    \begin{equation}
        v(X_G\mid X)=v(Y_G\mid Y)=\lambda(G\cap S).
    \end{equation}
\end{lemma}

\begin{proof}
    Since $X$ is $\cL^G$-invariant, we have that
    \begin{equation}
        \begin{multlined}
            v(X_G\mid X)
            =\sum_{x\in X_G}\int_S\lambda(\dd{s})\phi_x(s\mid X)
            =\sum_{x\in X_G}\sum_{t\in\cL^G}\int_{(t+G)\cap S}\lambda(\dd{s})\phi_x(s\mid X)
            \\
            =\sum_{x\in X_G}\sum_{t\in\cL^G}\int_{G\cap S}\lambda(\dd{s})\phi_{-t+x}(s\mid -t+X)
            =\int_{G\cap S}\lambda(\dd{s})\sum_{x\in X}\phi_x(s\mid X)
            =\lambda(G\cap S),
        \end{multlined}
    \end{equation}
    where we used Lemma~\ref{lem:translation invariance and lower semicontinuity} in the third equality and~\eqref{eqn:partition of unity} in the last.
    The same argument applies to $Y$.
\end{proof}

We now establish an important screening property of the allocated volume.

\begin{lemma}
    \label{lem:allocated volume screening}
    For every bounded, $\cF$-measurable region $\Lambda\subset\R^d$, perfect configuration $\overline{X}$ and admissible configuration $X$ coinciding with $\overline{X}$ on $\partial_F\Lambda\cup\Lambda^c$,
    \begin{equation}
        v(X_\Lambda\mid X)=v(\overline{X}_\Lambda\mid\overline{X}).
    \end{equation}
\end{lemma}

\begin{proof}
    By~\eqref{eqn:partition of unity},
    \begin{equation}
        \label{eqn:lem:allocated volume screening - splitting}
        \sum_{x\in X_\Lambda}\phi_x(\cdot\mid X)
        +\sum_{x\in X_{\Lambda^c}}\phi_x(\cdot\mid X)
        =1
        =\sum_{x\in\overline{X}_\Lambda}\phi_x(\cdot\mid\overline{X})
        +\sum_{x\in\overline{X}_{\Lambda^c}}\phi_x(\cdot\mid\overline{X})
        \quad\lambda\text{-a.e.}
    \end{equation}
    By Assumption~\ref{itm:volume screening}, the second sums on both sides of~\eqref{eqn:lem:allocated volume screening - splitting} are $\lambda$-a.e. equal.
    After cancellation, integrating both sides over $S$ yields the lemma.
\end{proof}

\begin{proof}[Proof of Proposition~\ref{prop:perfect iff periodic ground state}]
    By Lemma~\ref{lem:translation invariance and lower semicontinuity}, every translation of a perfect configuration is again a perfect configuration.
    Since there are only finitely many perfect configurations, each perfect configuration must be periodic.
    To see that a perfect configuration $\overline{X}$ is a ground state, let $\Lambda\subset\bG$ be bounded, and $X$ be an admissible configuration coinciding with $\overline{X}$ on $\Lambda^c$.
    Let $[\Lambda]$ be the smallest $\cF$-measurable set such that $\Lambda$ is contained in $[\Lambda]\setminus\partial_F[\Lambda]$.
    Then,
    \begin{equation}
        H(X_\Lambda\mid \overline{X}_{\Lambda^c})-H(\overline{X}_\Lambda\mid \overline{X}_{\Lambda^c})
        =-\sum_{x\in X_{[\Lambda]}}\mu_x+\sum_{x\in\overline{X}_{[\Lambda]}}\mu_x
        \ge -p_\ast\sum_{x\in X_{[\Lambda]}}v(x\mid X)+p_\ast\sum_{x\in\overline{X}_{[\Lambda]}}v(x\mid\overline{X})
        =0.
    \end{equation}
    where we used Assumption~\ref{asm:optimization} in the inequality and Lemma~\ref{lem:allocated volume screening} in the last equality.

    For the converse, let $Y$ be a periodic ground state. 
    It is clear that $Y$ must be non-empty.
    Suppose by contradiction that $Y$ is not a perfect configuration, namely, there exist $y\in Y$ and $\delta>0$ such that $p_\ast\, v(y\mid Y)=\mu_y+\delta$.
    Let $\overline{X}$ be a perfect configuration, and $G$ be a supercell larger than the interaction radius and common to both $Y$ and $\overline{X}$.
    By Assumption~\ref{asm:optimization},
    \begin{equation}
        H(\overline{X}_{\onion[n]{G}}\mid Y_{\R^d\setminus\onion[n+1]{G}})
        =-(2n+1)^d\sum_{x\in \overline{X}_G}\mu_x
        =-(2n+1)^d\,p_\ast\sum_{x\in\overline{X}_G}v(x\mid\overline{X}),
    \end{equation}
    while
    \begin{equation}
        H(Y_{\onion[n+1]{G}}\mid Y_{\R^d\setminus\onion[n+1]{G}})
        =-(2n+3)^d\sum_{y\in Y_G}\mu_y
        \ge-(2n+3)^d\,\left[p_\ast\sum_{y\in Y_G}v(y\mid Y)-\delta\right].
    \end{equation}
    By Lemma~\ref{lem:common supercell equal volume},
    \begin{equation}
        \begin{multlined}
            \limsup_{n\to\infty}\frac{1}{(2n+3)^d}\left[H(\overline{X}_{\onion[n]{G}}\mid Y_{\R^d\setminus\onion[n+1]{G}})-H(Y_{\onion[n+1]{G}}\mid Y_{\R^d\setminus\onion[n+1]{G}})\right]
            \\
            \le\limsup_{n\to\infty}\left[-\frac{(2n+1)^d}{(2n+3)^d}\,p_\ast\, v(\overline{X}_G\mid\overline{X})+p_\ast\, v(Y_G\mid Y)\right]-\delta
            =-\delta
            <0,
        \end{multlined}
    \end{equation}
    contradicting the assumption that $Y$ is a ground state.
\end{proof}

\subsubsection{The Peierls condition}

In this sub-subsection, we verify the Peierls condition under Assumptions~\ref{asm:optimization} and~\ref{asm:supercell}.
We assume familiarity with the standard contour formalism as presented in~\cite{zahradnik1984alternate}, including the notation $\supp{\Gamma}$ and $\intr{\Gamma}$ for the support and interior of a contour $\Gamma$, and $V(\Gamma):=\supp{\Gamma}\cup\intr{\Gamma}$.
Since all of these sets are $\cF$-measurable, by the equivalence $\cL^F\cong\Z^d$, we will use $\abs{\cdot}$ to denote the number of cells in a given $\cF$-measurable set.

\begin{proposition}
    \label{prop:Peierls condition}
    The \textbf{Peierls condition} holds: There exists a constant $\tau>0$ such that for every contour $\Gamma=(\supp{\Gamma},X_{\supp{\Gamma}})$,
    \begin{equation}
        \label{eqn:Peierls condition}
        H(X_{\supp{\Gamma}}\mid X_{\supp{\Gamma}^c})-H(\overline{X}_{\supp{\Gamma}}\mid \overline{X}_{\supp{\Gamma}^c})
        \ge \tau\abs{\supp{\Gamma}}.
    \end{equation}
\end{proposition}

The main ingredient in the proof of Proposition~\ref{prop:Peierls condition} is the following lemma, which guarantees a minimum increase in the allocated volume within a uniformly bounded neighborhood of every incorrect cell.
Similar ideas first appeared in~\cite{holsztynski1978peierls}.

\begin{lemma}
    \label{lem:defect nearby}
    There exist constants $\delta>0$ and $n\ge 1$ such that for every cell $K\in\cT^F$ and admissible configuration $X\in\Omega$ with respect to which $K$ is incorrect, either $X_{\onion[1]{K}}=\emptyset$, or there exists $x\in X_{\onion[n]{K}}$ such that $p_\ast\, v(x\mid X)\ge\mu_x+\delta$.
\end{lemma}

\begin{proof}
    Without loss of generality, we may assume that $K=F$.
    By contradiction, there exists a sequence of admissible configurations $(X_n)_{n\ge 1}$ such that for every $n\ge 1$, (a) $F$ is incorrect with respect to $X_n$, (b) $X_n\cap \onion[1]{F}\ne\emptyset$, and (c) for every $x\in X_n\cap \onion[n]{F}$, $p_\ast\, v(x\mid X_n)<\mu_x+\frac{1}{n}$.
    By passing to a subsequence if necessary, we may further assume that there exists an admissible configuration $X\in\Omega$ such that $X_n$ converges to $X$.
    Since $X_n\cap \onion[1]{F}\ne\emptyset$ for every $n\ge 1$, it holds that $X\cap \onion[1]{F}\ne\emptyset$, so, in particular, $X$ is non-empty.
    Let $x\in X$.
    Then, for every sufficiently large $n$, $x\in X_n\cap \onion[n]{F}$ and thus by Lemma~\ref{lem:translation invariance and lower semicontinuity},
    \begin{equation}
        \mu_x
        \le p_\ast\, v(x\mid X)
        \le \liminf_{n\to\infty}p_\ast\, v(x\mid X_n)
        \le \liminf_{n\to\infty}\left(\mu_x+\frac{1}{n}\right)
        =\mu_x.
    \end{equation}
    Hence, $X$ is a perfect configuration.
    However, for every sufficiently large $n$, $X_n$ coincides with $X$ on $\onion[1]{F}$, and so $F$ is correct with respect to $X_n$, which is a contradiction.
\end{proof}

We also need the following identity for the total allocated volume in the support of a contour.

\begin{lemma}
    \label{lem:equivalence to lattice volume} 
    For every contour $\Gamma=(\supp{\Gamma},X_{\supp{\Gamma}})$, 
    \begin{equation}
        v(X_{\supp{\Gamma}}\mid X)=\lambda(F\cap S)\abs{\supp{\Gamma}}.
    \end{equation}
\end{lemma}

\begin{proof}
    Let $\overline{X}$ be the perfect configuration such that $X_{\partial_\cF V(\Gamma)}=\overline{X}_{\partial_\cF V(\Gamma)}$.
    We compute
    \begin{equation}
        \begin{multlined}
            v(X_{\supp{\Gamma}}\mid X)
            =v(X_{V(\Gamma)}\mid X)-v(X_{\intr{\Gamma}}\mid X)
            \\
            =v(\overline{X}_{V(\Gamma)}\mid\overline{X})-v(X_{\intr{\Gamma}}\mid X)
            =\lambda(F\cap S)\abs{\supp{\Gamma}},
        \end{multlined}
    \end{equation}
    where we used Lemmas~\ref{lem:allocated volume screening} and~\ref{lem:common supercell equal volume} in the second and third equalities, respectively.
\end{proof}

\begin{proof}[Proof of Proposition~\ref{prop:Peierls condition}]
    Recall that $\mu_1,\dots,\mu_k\ge 0$.
    If there exists a cell $K\in\cT^F$ contained in $\supp{\Gamma}$ such that $X_{\supp{\Gamma}}\cap \onion[1]{K}=\emptyset$, then by~\eqref{eqn:interaction screening}, we can add a particle to $X_{\supp{\Gamma}}$ in $K$ to obtain a new admissible configuration with equal or less energy.
    Repeating the argument as needed, we may assume, for the purpose of proving~\eqref{eqn:Peierls condition}, that $\supp{\Gamma}$ does not contain any such $K$.
    Thus, by Lemma~\ref{lem:defect nearby}, there exist universal constants $\delta>0$ and $n\ge 1$ such that for every cell $K\in\cT^F$ contained in $\supp{\Gamma}$, there exists $x\in X_{\onion[n]{K}}$ with $p_\ast\, v(x\mid X)\ge\mu_x+\delta$; in fact, $x\in X_{\supp{\Gamma}}\cap \onion[n]{K}$ since every cell $K\in\cT^F$ contained in $\supp{\Gamma}^c$ is correct with respect to $X$.

    Now, by Assumption~\ref{asm:optimization},
    \begin{equation}
        \label{eqn:lower bound on contour energy}
        H(X_{\supp{\Gamma}}\mid X_{\supp{\Gamma}^c})
        =-\sum_{x\in X_{\supp{\Gamma}}}\mu_x
        \ge -p_\ast\sum_{x\in X_{\supp{\Gamma}}}v(x\mid X)+\frac{\delta}{(2n+1)^d}\abs{\supp{\Gamma}},
    \end{equation}
    and
    \begin{equation}
        \label{eqn:identity for ground state contour energy}
        H(\overline{X}_{\supp{\Gamma}}\mid \overline{X}_{\supp{\Gamma}^c})
        =-\sum_{x\in \overline{X}_{\supp{\Gamma}}}\mu_x
        =-p_\ast\sum_{x\in \overline{X}_{\supp{\Gamma}}}v(x\mid \overline{X}).
    \end{equation}
    Combining~\eqref{eqn:lower bound on contour energy},~\eqref{eqn:identity for ground state contour energy}, and Lemma~\ref{lem:equivalence to lattice volume}, we have that
    \begin{equation}
        H(X_{\supp{\Gamma}}\mid X_{\supp{\Gamma}^c})-H(\overline{X}_{\supp{\Gamma}}\mid \overline{X}_{\supp{\Gamma}^c})
        \ge\frac{\delta}{(2n+1)^d}\abs{\supp{\Gamma}},
    \end{equation}
    so~\eqref{eqn:Peierls condition} holds with $\tau=\delta/(2n+1)^d$.
\end{proof}

\subsubsection{Deduction of Theorem~\ref{thm:main}}

Among the Richness, Model, and Peierls Assumptions of~\cite{mazel2024pirogov}, only Item~(iv) of the Model Assumption and Item~(iv) of the Peierls Assumption require verification.
The former requires the existence of only finitely many periodic ground states and follows from Assumption~\ref{itm:global optimization} and Proposition~\ref{prop:perfect iff periodic ground state}.
The latter is the Peierls condition for contours, which is already proven in Proposition~\ref{prop:Peierls condition}.

\section{Applications}
\label{sec:applications}

\subsection{Construction of volume allocations}

\subsubsection{Discrete Voronoi tessellation}

In this sub-subsection, we show that the discrete Voronoi tessellation considered in~\cite{he2024high} naturally gives rise to a volume allocation.

Let $\lambda$ denote the counting measure on $\bG$ and $\mathrm{d}$ be an $\bL$-invariant metric on $\bG$ such that every ball has finite $\lambda$-measure.
Let $X\in\Omega$ be a non-empty admissible configuration and $x\in X$.
Recall the notation $T_x$ introduced in Section~\ref{sec:the model}.
Write
\begin{equation}
    \label{eqn:discretized tile}
    D_x:=T_x\cap\bG
\end{equation}
(which we assume always to be non-empty), and define the \textbf{discrete Voronoi cell} of the discretized tile $D_x$ with respect to $X$ as the set
\begin{equation}
    V_X(D_x):=\set{v\in\bG\mid \mathrm{d}(v,D_x)\le \mathrm{d}(v,D_y)\text{ for every }y\in X}.
\end{equation}
For $v\in\bG$, define
\begin{equation}
    \label{eqn:volume allocation via discrete Voronoi tessellation}
    \phi_x(v\mid X):=\frac{\indicator{V_X(D_x)}(v)}{\sum_{z\in X}\indicator{V_X(D_z)}(v)}.
\end{equation}

\begin{lemma}
    \label{lem:discrete Voronoi volume allocation}
    The family of functions~\eqref{eqn:volume allocation via discrete Voronoi tessellation} defines a volume allocation.
\end{lemma}

\begin{proof}
    We will only verify the lower semi-continuity property~\eqref{eqn:lower semi-continuity}, as the other properties are immediate.
    Let $v\in\bG$, $X$ be a non-empty admissible configuration, $x\in X$, and $(X_n)_{n\ge 1}$ be a sequence of admissible configurations converging to $X$.
    Without loss of generality, we may assume that $x\in X_n$ for every $n\ge 1$.
    Observe that $\phi_x(v\mid X_n)$ is completely determined by the restriction of $X_n$ to the (fixed) finite set
    \begin{equation}
        \set{u\in\bG\mid\text{there exists }1\le i\le k\text{ such that }\mathrm{d}(v,(u+T_i)\cap\bG)\le\mathrm{d}(v,D_x)}.
    \end{equation}
    By the definition of local convergence, for every sufficiently large $n$, the restriction of $X_n$ to this set coincides with that of $X$, so that $\phi_x(v\mid X_n)=\phi_x(v\mid X)$.
\end{proof}

\subsubsection{Voronoi tessellation}

For tile shapes that do not have well-behaved discretizations, it may be more convenient to consider the usual Voronoi tessellation in $\R^d$.

Let $\lambda$ denote the Lebesgue measure on $\R^d$ and $\mathrm{d}$ be the Euclidean metric. 
Given a non-empty admissible configuration $X\in\Omega$ and $x\in X$, define the \textbf{Voronoi cell} of the tile $T_x$ with respect to $X$ as the set
\begin{equation}
    \label{eqn:Voronoi cell}
    V_X(T_x):=\set{v\in\R^d\mid \mathrm{d}(v,T_x)\le\mathrm{d}(v,T_y)\text{ for every }y\in X}.
\end{equation}
For $v\in\R^d$, define
\begin{equation}
    \label{eqn:volume allocation via Voronoi tessellation}
    \phi_x(v\mid X):=\frac{\indicator{V_X(T_x)}(v)}{\sum_{z\in X}\indicator{V_X(T_z)}(v)}.
\end{equation}
A completely analogous argument to the proof of Lemma~\ref{lem:discrete Voronoi volume allocation} shows that:

\begin{lemma}
    \label{lem:Voronoi volume allocation}
    The family of functions~\eqref{eqn:volume allocation via Voronoi tessellation} defines a volume allocation.
\end{lemma}

\subsubsection{Other tessellations}

The discrete and continuous Voronoi tessellations mentioned above are but two examples of tessellations that can be used to construct volume allocations.
A treatment of other common tessellations in the context of stochastic geometry may be found in~\cite[Chapter~9]{chiu2013stochastic}.
It is plausible that many can be adapted to define volume allocations, providing flexibility for models with specific geometric features and constraints.
A notable example is the Delaunay tessellation (triangulation), which Mazel--Stuhl--Suhov~\cite{mazel2025high} successfully used to characterize defects in dense packings of identical disks in the context of Pirogov--Sinai theory.
While we do not attempt to verify whether their notion of a \textit{redistributed area for Delaunay triangles} fits into our framework, their work highlights the utility of alternative tessellations in the rigorous analysis of specific lattice hard-core models.

\subsection{Comparison with Jauslin--Lebowitz and He--Jauslin}

In this subsection, we show that Theorem~\ref{thm:main} strictly generalizes the crystallization criteria in~\cite{Jauslin2018} and~\cite{he2024high} by proving that the main assumptions in either paper imply Assumptions~\ref{asm:optimization} and~\ref{asm:supercell} above.
Throughout, we use the volume allocation~\eqref{eqn:volume allocation via discrete Voronoi tessellation}.

We start with~\cite{Jauslin2018}.
The setup there consisted of a lattice $\Lambda_\infty$ equipped with a metric $\Delta$ and a single particle shape $\omega\subset\R^d$, assumed to be connected and bounded.
Thus, their model corresponds to taking $\bL=\bG=\Lambda_\infty$, $T_1=\omega$, and $\mu_1=1$.
They further introduced discretized tiles as in~\eqref{eqn:discretized tile}, and defined a \emph{perfect covering} as an admissible configuration such that the union of the discretized tiles coincides with $\Lambda_\infty$.
In their main \emph{non-sliding} assumption~\cite[Definition 1]{Jauslin2018}, they postulated the existence of finitely many periodic perfect coverings related by isometries of $\Lambda_\infty$.
It follows from the definition after~\cite[(24)]{Jauslin2018} that these are the \emph{only} periodic perfect coverings by applying the definition to a singleton configuration $X$.
Hence, Assumption~\ref{asm:optimization} holds with $p_\ast$ equal to the inverse of the cardinality of a discretized tile; in particular, the perfect configurations of the model are precisely the perfect coverings.
Assumption~\ref{itm:volume screening} is clear from the definition of a perfect covering.

We now turn to~\cite{he2024high}.
The setup there consisted of a periodic graph $\Lambda_\infty$ equipped with a metric $\mathrm{d}_{\Lambda_\infty}$ and a single bounded particle shape $\omega\subset\R^d$.
They introduced the same discretized tiles as above, but continued to formulate much of their main assumption~\cite[Assumption 1]{he2024high} in terms of the discrete Voronoi cells of the discretized tiles.
This greatly simplifies the comparison with our framework due to direct correspondences between parts of their assumption and ours.
Indeed, Assumption~\ref{itm:minorization} follows from~\cite[Assumption~1.4]{he2024high}, \ref{itm:global optimization} from~\cite[Lemma~2.12]{he2024high}, and \ref{itm:volume screening} from the uniform boundedness of the discrete Voronoi cells in the perfect configurations~\cite[Lemma~3.3]{he2024high} by an analogous argument as in the proof of Lemma~\ref{lem:discrete Voronoi volume allocation}.

\subsection{Polyominoes with a finite number of tilings}
\label{sec:polyominoes}

Polyominoes are a rich class of two-dimensional tiles whose packing properties have been extensively studied in the literature~\cite{golomb1996polyominoes,grunbaum1987tilings}.
They are defined as finite unions of unit squares centered on $\Z^2$ such that every two squares in the polyomino are connected by a path of edge-adjacent squares.
Originally popularized as combinatorial puzzles, polyominoes have found significant applications in the sciences as the simplest models of self-assembly~\cite{barnes2009structure,barnes2010development,shah2022phase}.

In these contexts, the high-fugacity behavior of polyomino fluids (and mixtures) is often closely related to the structure of the maximum-density packings.
However, to say nothing of identifying the densest packings, determining whether an arbitrary set of polyominoes can tile the plane is already \textit{computationally undecidable}~\cite{golomb1970tiling}, though some general criteria such as the extension theorem~\cite[Theorem 3.8.1]{grunbaum1987tilings} exist, and an exhaustive enumeration procedure known as \textit{backtracking} can be used to systematically search for tilings~\cite{golomb1996polyominoes,martin1991polyominoes}.
The latter, for instance, allows one to check that the polyominoes in Figures~\ref{fig:Z pentomino crystals} and~\ref{fig:other shapes} admit the indicated number of tilings.
Following~\cite{martin1991polyominoes}, we designate a polyomino as \textbf{$p$-poic}, or simply \textbf{polypoic}, if it admits $p$ distinct tilings up to translations and rotations, and \textbf{$m$-morphic}, or simply \textbf{polymorphic}, if it admits $m$ distinct tilings up to translations, rotations, and reflections.

Let $T$ be a polypoic or polymorphic polyomino.
Following the notation of Section~\ref{sec:the model}, we set $\bL=\bG=\Z^2$, and let $T_1,\dots,T_k$ be the translationally inequivalent images of $T$ under rotations, and reflections in the polymorphic case, of $\Z^2$, such that the leftmost unit square in the bottom row of each $T_i$ is centered at the origin.
We associate to each $T_i$ the same chemical potential by setting $\vec{\mu}=\frac{1}{\sqrt{k}}(1,\dots,1)$.
Hence, for chiral polyominoes, we are considering the uniform mixture of both enantiomers.
The main result of this subsection is as follows.

\begin{corollary}
    \label{cor:polyominoes}
    On $\Z^2$, every polypoic polyomino and the uniform mixture of every chiral polymorphic polyomino crystallize at low temperatures in the sense of Theorem~\ref{thm:main}.
\end{corollary}

\begin{proof}
    We use the volume allocation~\eqref{eqn:volume allocation via discrete Voronoi tessellation}.
    Assumption~\ref{asm:optimization} holds with $p_\ast=1/(\sqrt{k}\Area{T})$ since the number of tilings is finite.
    Assumption~\ref{asm:supercell} follows by an analogous argument as in the proof of Lemma~\ref{lem:discrete Voronoi volume allocation}.
    The corollary now follows from Theorem~\ref{thm:main}.
\end{proof}

\begin{remark}
    Corollary~\ref{cor:polyominoes} admits immediate extensions to polyiamond, polyhex, and polycube fluids under the same condition that the number of tilings is finite.
    For tiles with infinitely many tilings, the high-fugacity behavior has been rigorously characterized in only a limited number of cases, including the domino~\cite{van1999absence,cohn2001variational,sheffield2005random} and the square tetromino~\cite{hadas2025columnar}.
    Tiles that do not admit any tilings can also be treated within our framework, though verifying the assumptions will require a (non-trivial) specialized analysis depending on the specific tile geometry; see~\cite{he2024high} for a few examples.
\end{remark}

\subsection{Diamond-octagon mixtures}
\label{sec:diamond-octagon mixtures}

The purpose of this subsection is to give a simple example of a \emph{bona fide} multi-component mixture, featuring visibly distinct tiles, that crystallizes at low temperatures.
Motivated by the extensive catalog of planar tilings in~\cite{grunbaum1987tilings}, we consider a diamond-octagon mixture that gives rise to the well-known \textit{truncated square tiling}; see Figure~\ref{fig:truncated square tiling}.
A lattice version of this system is easily formulated within our framework by setting $\bL=\bG=\Z^2$, $T_1$ as the tetragon with vertices $(\pm1,0),(0,\pm1)$, and $T_2$ as the octagon with vertices $(0,-1),(1,-1),(2,0),(2,1),(1,2),(0,2),(-1,1),(-1,0)$.

\begin{figure}[t]
    \centering
    \begin{subfigure}{0.3\columnwidth}
        \centering
        \includegraphics[scale=0.5]{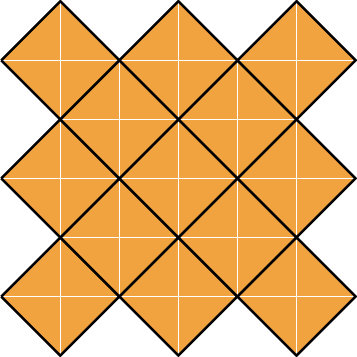}
        \caption{}
    \end{subfigure}
    \begin{subfigure}{0.3\columnwidth}
        \centering
        \includegraphics[scale=0.5]{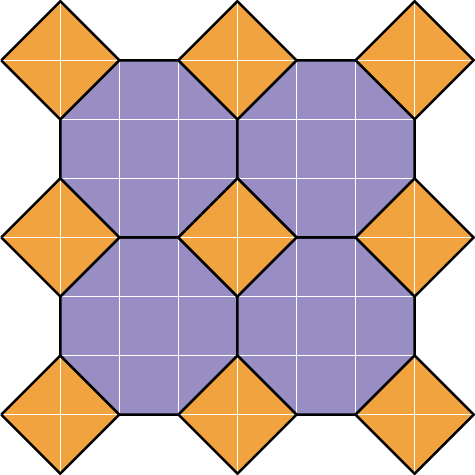}
        \caption{}
        \label{fig:truncated square tiling}
    \end{subfigure}
    \caption{The only crystal (tiling) structures of the diamond-octagon mixture on the square lattice up to translations, given the ratio of (positive) chemical potentials $\mu_{\text{diamond}}:\mu_{\text{octagon}}=2:7$.}
    \label{fig:diamond-octagon mixture}
\end{figure}

\begin{proposition}
    If $(\mu_1,\mu_2)=\frac{1}{\sqrt{53}}(2,7)$, then Assumptions~\ref{asm:optimization} and~\ref{asm:supercell} hold for the diamond-octagon mixture with~\eqref{eqn:volume allocation via Voronoi tessellation} as the volume allocation and $p_\ast=1/\sqrt{53}$.
    Consequently, the model crystallizes at low temperatures in the sense of Theorem~\ref{thm:main}.
    In particular, the model admits exactly two classes of perfect configurations up to translations of $\Z^2$, one given by the truncated square tiling, and the other by the regular square tiling with diamonds only; see Figure~\ref{fig:diamond-octagon mixture}.
\end{proposition}

\begin{proof}
    Assumption~\ref{itm:minorization} is clear from the definition~\eqref{eqn:Voronoi cell} of a Voronoi cell, which has at least as much volume as the tile itself.
    For Assumption~\ref{itm:global optimization}, we seek admissible configurations in which every diamond has allocated volume exactly $2$ and every octagon exactly $7$.
    If there is no octagon in the configuration, then the only way to achieve this is to have a regular square tiling by diamonds.
    Hence, we may assume that the configuration contains at least one octagon, without loss of generality, located at the origin and denoted by $O$.
    We claim that the four triangular connected components of $[-1,2]^2\setminus O$ must each be covered by a diamond. 
    Indeed, if any of these components is covered by another octagon $O'$, then $O$ and $O'$ border a common unit square of $\Z^2$ that cannot be completely covered by a diamond or an octagon, so the Voronoi cell of $O$ has strictly greater volume than $7$, which is forbidden.
    We have thus accounted for the diagonal edges of $O$.
    Next, it is easy to see that each vertical and horizontal edge of $O$ must border another octagon. 
    Repeating the above argument for the adjacent octagons, we conclude that the configuration must be a truncated square tiling.
    Assumption~\ref{itm:volume screening} is now clear.
\end{proof}

\subsection*{Acknowledgments}

We thank Izabella Stuhl for answering a question about the statement of~\cite[Model Assumptions (i)]{mazel2024pirogov} and Ian Jauslin, Joel Lebowitz, and Ron Peled for insightful questions about the work during the author's dissertation defense that led to various improvements of the note.
The author is supported by an SAS Fellowship at Rutgers University.

\bibliographystyle{plain}
\bibliography{bibliography}

\end{document}